\documentclass[letterpaper]{article}
\usepackage{graphicx} 
\usepackage[margin=1in]{geometry}
\usepackage{algorithm}
\usepackage{algpseudocode}
\usepackage{amsthm}
\usepackage{amssymb,amsmath}  
\usepackage{mathtools}
\usepackage[colorlinks=true, allcolors=blue]{hyperref}
\usepackage[capitalize]{cleveref}
\usepackage{todonotes}
\usepackage{placeins}
\usepackage{subcaption}
\usepackage{booktabs}
\usepackage{enumitem}
\usepackage{xcolor}
\usepackage{authblk}

\def\dd{\mathinner{.\,.}}
\newcommand{\cO}{\mathcal{O}}
\newcommand{\MP}{\mathcal{M}_F}
\newcommand{\PP}{\mathcal{P}_F}
\newcommand{\ctO}{\mathcal{\widetilde{O}}}

\title{Maximal Palindromes in MPC: Simple and Optimal}
\author[1,2]{Solon P.\ Pissis}
\affil[1]{CWI, Amsterdam, The Netherlands}
\affil[2]{Vrije Universiteit, Amsterdam, The Netherlands}
\date{\today}

\newtheorem{theorem}{Theorem}[section]
\newtheorem{corollary}[theorem]{Corollary}
\newtheorem{lemma}[theorem]{Lemma}

\newtheorem{fact}[theorem]{Fact}

\newtheorem{remark}[theorem]{Remark}

\begin{document}

\maketitle

\begin{abstract}
In the classical \emph{longest palindromic substring} (LPS) problem, we are given a string $S$ of length $n$, and the task is to output a longest palindromic substring in $S$.
Gilbert, Hajiaghayi, Saleh, and Seddighin [SPAA 2023] showed
how to solve the LPS problem in the \emph{Massively Parallel Computation} (MPC) model in $\cO(1)$ rounds using $\ctO(n)$ total memory, with $\ctO(n^{1-\epsilon})$ memory per machine, for any $\epsilon \in (0,0.5]$.

We present a simple and optimal algorithm to solve the LPS problem in
the MPC model in $\cO(1)$ rounds. The total time and memory are $\cO(n)$, with $\cO(n^{1-\epsilon})$ memory per machine, for any $\epsilon \in (0,0.5]$. A key attribute of our algorithm is its ability to compute all \emph{maximal palindromes} in the same complexities. Furthermore, our new insights allow us to bypass the constraint $\epsilon \in (0,0.5]$ in the \emph{Adaptive} MPC model.
Our algorithms and the one proposed by Gilbert et al. for the LPS problem are randomized and succeed with high probability.
\end{abstract}

\section{Introduction}\label{sec:intro}
 
In the classical \emph{longest palindromic substring} (LPS) problem, we are given a string $S$ of length $n$ over an alphabet $\Sigma$, and the task is to output a longest palindromic substring in $S$. The LPS problem can be solved efficiently. Several algorithms achieve $\cO(n)$-time complexity, including Manacher's celebrated algorithm~\cite{DBLP:journals/jacm/Manacher75,DBLP:journals/tcs/ApostolicoBG95}, Jeuring's algorithm~\cite{DBLP:journals/algorithmica/Jeuring94}, and Gusfield's algorithm, which leverages longest common prefix queries~\cite{DBLP:books/cu/Gusfield1997}. In fact, these algorithms can output all \emph{maximal palindromes}: the longest palindrome centered at every position within $S$. More recently, Charalampopoulos, Pissis, and Radoszewski~\cite{DBLP:conf/cpm/Charalampopoulos22} gave an algorithm that solves LPS in $\cO(n\frac{\log \sigma}{\log n})$ time, where $\Sigma=[0,\sigma)$ is an integer alphabet. If, for instance, $\Sigma=\{0,1\}$ and $\sigma=2$, their algorithm runs in $\cO(n/\log n)$ time.
The LPS problem has also been studied in many other settings, such as the compressed setting, where the string $S$ is given as a straight-line program~\cite{DBLP:journals/tcs/MatsubaraIISNH09}, the streaming setting~\cite{DBLP:journals/algorithmica/GawrychowskiMSU19}, the dynamic setting, where $S$ undergoes edit operations~\cite{DBLP:journals/algorithmica/AmirCPR20}, and a semi-dynamic setting~\cite{DBLP:journals/tcs/FunakoshiNIBT21}. In the quantum setting, Le Gall and Seddighin~\cite{DBLP:journals/algorithmica/GallS23} gave a strongly sublinear-time algorithm complemented with a lower bound.

In this work, we study the LPS problem in the \emph{Massively Parallel Computation} (MPC) model~\cite{DBLP:conf/soda/KarloffSV10,DBLP:conf/isaac/GoodrichSZ11,DBLP:conf/stoc/AndoniNOY14,DBLP:journals/jacm/BeameKS17,DBLP:journals/ftopt/ImKLMV23}. In MPC, problems with a data size of $\cO(n)$ are distributed across multiple machines, each having a strongly sublinear memory. Algorithms in MPC operate in a sequence of rounds. During each round, the machines independently perform computations on their local data. Following this, the machines communicate with each other. Given that communication often presents a significant performance bottleneck in real-world applications, the primary objective when designing MPC algorithms is to minimize the round complexity while ensuring that the total memory across all machines remains $\cO(n)$. Specifically, each machine has $\cO(n^{1-\epsilon})$ memory, where $\epsilon \in (0,1)$. During a single communication round, machines can send and receive any number of messages, provided that the total size of these messages fits within their local memory limits. Typically, we require that each machine's local memory is sufficient to store a message from every other machine in a single communication round. To facilitate this, we must ensure that a machine's local memory, $\cO(n^{1-\epsilon})$, is greater than the number of machines, which is $\cO(n^{\epsilon})$. We thus restrict $\epsilon$ to $\epsilon \in (0,0.5]$. This is justified within the MPC model because, with current computing resources, the number of machines typically does not exceed the local memory capacity of each machine. 

We also consider the \emph{Adaptive} MPC (AMPC) model~\cite{DBLP:journals/topc/BehnezhadDELMS21}, which extends the MPC model by allowing machines to access a shared read-only memory \emph{within a round}. This is modeled by writing all messages sent in round $i-1$ to a distributed data storage, which all machines can read from within round $i$.

\paragraph{State of the Art.} Gilbert, Hajiaghayi, Saleh, and Seddighin~\cite{DBLP:conf/spaa/GilbertHSS23} showed how to solve the LPS problem in the MPC model in $\cO(1)$ rounds using $\ctO(n)$ total memory, with $\ctO(n^{1-\epsilon})$ memory per machine, for any $\epsilon \in (0,0.5]$, with high probability (w.h.p). Their algorithm is based on periodicity arguments borrowed from~\cite{DBLP:journals/tcs/ApostolicoBG95} and a powerful oracle they introduce to answer longest common prefix queries. Unfortunately, it was difficult to assess the simplicity of their algorithm, as many proofs (and subroutines) are deferred to the full version of their work. For the same reason, we have been unable to count the polylogarithmic factors in the claimed $\ctO(n)$ total memory. Finally, the authors make no claim about the total running time of their algorithm.

\paragraph{Our Contributions and Paper Organization.}~Our central contribution is a simple and optimal algorithm to solve the LPS problem in the MPC model in $\cO(1)$ rounds. The total time and memory are $\cO(n)$, with $\cO(n^{1-\epsilon})$ memory per machine, for any $\epsilon \in (0,0.5]$, w.h.p. Our algorithm translates several of the combinatorial insights of Charalampopoulos et al.~\cite{DBLP:conf/cpm/Charalampopoulos22} into a suitable block decomposition for the MPC model, and carefully combines the latter with a modular decomposition, which was also utilized by Gilbert et al. in~\cite{DBLP:conf/spaa/GilbertHSS23}. 
A key attribute of our algorithm is its ability to compute all \emph{maximal palindromes} in the same complexities. Furthermore, our new insights allow us to bypass the constraint $\epsilon \in (0,0.5]$ in the stronger AMPC model. Our algorithms can be implemented by anyone with a basic knowledge of parallel programming using our $\tilde{3}$-page description, which also includes the arguments to fully verify the correctness of our algorithms. In \Cref{sec:prel}, we present some basic concepts and in \Cref{sec:algo}, we present our algorithms. We conclude the paper in \Cref{sec:con} with some final remarks.

\section{Basics}\label{sec:prel}

\paragraph{Strings.}~We consider an integer alphabet $\Sigma=[0,\sigma)$ of size $\sigma$. A \emph{string} $S = S[0] \dots S[n-1]$ is a sequence over $\Sigma$;
its \emph{length} is denoted by $|S| = n$. For $0 \leq i \leq j < n$, 
a string $S[i] \dots S[j]$ is called a \emph{substring} of $S$. 
By $S[i\dd j]$, we denote its occurrence at position $i$, which is called a \emph{fragment} of $S$. A fragment with $i = 0$ is called a \emph{prefix} (also denoted by $S[\dd j]$) and a fragment with $j = n-1$ is called a \emph{suffix} (also denoted by $S[i\dd ]$). By $ST$ or $S\cdot T$, we denote the \emph{concatenation} of
two strings $S$ and $T$. We denote the \emph{reverse} string of $S$ by $S^R$, 
i.e., $S^R = S[n-1] \dots S[0]$. The string $S$ is a \emph{palindrome} if and only if $S = S^R$. If $S[i\dd j]$ is a palindrome, the number $\frac{i+j}{2}$ is called the \emph{center} of $S[i\dd j]$. A palindromic fragment $S[i \dd j]$ of $S$ is said to be a \emph{maximal palindrome} if there is no longer palindrome in $S$ with center $\frac{i+j}{2}$. Note that a longest palindromic substring in $S$ must be maximal. A positive integer $p$ is called a \emph{period} of a string $S$, if $S[i] = S[i + p]$, for all $i \in [0, |S| - p)$.

\begin{fact}[cf.~\cite{DBLP:journals/jda/FiciGKK14}]\label{fct:per-pal}
Let $U$ be a prefix of a palindrome $V$, with $|U|<|V|$. Then $|V|-|U|$
is a period of $V$ if and only if $U$ is a palindrome. 
In particular, $|V|-|U|$ is the smallest period of $V$ if and
only if $U$, with $U \neq V$, is the longest palindromic prefix of $V$.
\end{fact}

\paragraph{LCP Queries.}~Let $S$ be a string of length $n$. For any two positions $i, j$ ($0 \leq i, j < n$) in $S$, we define $\textsf{LCP}_S(i, j)$ as the length of the longest common prefix (LCP) of the suffixes $S[i\dd]$ and $S[j\dd ]$. It is well known that computing the maximal palindrome of any center $c$ in $S$ reduces to asking $\textsf{LCP}_{S'}(c,2n-c-1)$ for odd-length palindromes and $\textsf{LCP}_{S'}(\lceil c \rceil,2n-\lceil c \rceil)$ for even-length palindromes, where $S':=S\cdot S^R$~\cite{DBLP:books/cu/Gusfield1997}.

\paragraph{Karp-Rabin Fingerprints.}~For a prime number $p$ and an integer $x \in \mathbb{Z}^+$, the \emph{Karp-Rabin} (KR) \emph{fingerprint}~\cite{DBLP:journals/ibmrd/KarpR87} 
of a string $S$ of length $n$ is $\phi(S)=(\sum^{n-1}_{i=0} S[i] \cdot x^{i}) \mod p$. We also maintain $(x^{n-1}\mod p,x^{-(n-1)}\mod p,n)$ with it for efficiency. The KR fingerprints for $S$ are \emph{collision-free} if $\phi(S[i\dd j]) = \phi(S[i'\dd j'])$ implies $S[i\dd j]=S[i'\dd j']$. Using randomization, we can construct such fingerprints succeeding \emph{with high probability} (w.h.p.).

\begin{fact}[cf.~\cite{DBLP:conf/stacs/IKK14}]\label{fct:KR}
Let $S$ be a string of length $n$ over an alphabet $\Sigma$, 
and let $p\geq \max(|\Sigma|,n^{3+c})$ be a prime number. If $x$ is chosen uniformly at random, then $\phi$ is collision-free with probability at least $1 - n^{-c}$.
\end{fact}

\begin{fact}[cf.~\cite{DBLP:conf/stacs/IKK14}]\label{fct:fastKR}
    Let $U, V, W$ be strings such that $UV = W$. 
    Given two of the three KR fingerprints $\phi(U), \phi(V), \phi(W)$, the third can be computed in $\cO(1)$ time.
\end{fact}

\section{The Algorithms}\label{sec:algo}

Let $S\in \Sigma^n$ be the input string to the LPS problem.
In addition to $S$, we will consider the string $S':=S\cdot S^R \in \Sigma^{2n}$ to be input to our algorithms. Clearly, the total input size is $\cO(n)$, as required.

\paragraph{Combinatorial Insights.} Let $F=S[i\dd j]=B_1B_2B_3B_4$ be a fragment of $S$ of length $\ell=4\ell'$, with $\ell'=|B_1|=|B_2|=|B_3|=|B_4|>0$. 
Let $\PP$ be the set of palindromes in $S$ that are \emph{prefixes} of $F$ with centers in $[i+\ell',i+2\ell')$ (i.e., in $B_2$). Further let $\MP$ be the set of \emph{maximal} palindromes in $S$ with centers in $[i+\ell',i+2\ell')$ that either \emph{exceed} $F$ or are \emph{prefixes} of $F$.   The following structural lemma shows that either $|\MP|\leq 1$ or the palindromes in $\PP$ share the same period. This lemma will inspire our \emph{block decomposition}.

\begin{lemma}\label{lem:structural}
Let $F=S[i\dd j]=B_1B_2B_3B_4$ be a fragment of $S$ of length $\ell=4\ell'$, with $\ell'=|B_1|=|B_2|=|B_3|=|B_4|>0$. Then the following hold: (1) If $|\PP|=0$, then $|\MP|=0$; (2) If $|\PP|=1$, then $|\MP|=1$; (3) 
If $|\PP|\geq 2$, every $P\in \PP$ has a period $p$, where $p$ is the smallest period of
the longest palindrome in $\PP$.
\end{lemma}

\begin{proof}
By the decomposition of $F$ into $B_1B_2B_3B_4$, every $M\in \MP$ has as a \emph{substring} a palindrome $P\in \PP$ with the same center (in $B_2$). Thus the first two items are immediate.

For $|\PP|\geq 2$, let us denote the longest palindrome in $\PP$ by $P_1$. 
Let $P_2$ denote any other palindrome in $\PP$. 
Let $n_1 := |P_1|$ and $n_2:=|P_2|$.
By \Cref{fct:per-pal}, $(n_1-n_2)$ is a period of $P_1$.
By definition, the smallest period $p$ of $P_1$ must be less than or equal to any other period of $P_1$. Thus, we have $p \leq n_1 - n_2$.

To show that $P_2$ has a period $p$, we must prove that $P_2[i] = P_2[i+p]$, for all $0 \leq i < n_2 - p$. Since $P_2$ is a prefix of $P_1$, the corresponding letters are identical for the entire length of $P_2$. By the definition of period, since $p$ is a period of $P_1$, we know that $P_1[i] = P_1[i+p]$, for all $0 \leq i < n_1 - p$.

From $n_2 < n_1$, we have $n_2 - p < n_1 - p$. If we assume that $n_2 - p$ is positive, this shows that the range of periodicity for $P_2$ is fully contained within the range where $P_1$'s periodicity is guaranteed. This means that the periodic property of $P_1$ holds over the entire length of its prefix $P_2$. Therefore, for any $i$ in the range $0 \leq i < n_2 - p$, we have
$P_2[i] = P_1[i] = P_1[i+p] = P_2[i+p]$. This proves that $P_2$ has a period $p$.

However, we still need to prove that $n_2 - p$ is positive; i.e., $n_2 > p$.
Recall $p \leq n_1 - n_2$.
The centers of $P_1$ and $P_2$ lie within $[\ell',2\ell')$ on $F$. This implies that the difference in their lengths is at most $2\ell'$: $n_1 - n_2 \leq 2\ell'$.
The center $c_2$ of $P_2$ is such that $c_2 \geq \ell'$. This gives a lower bound for the length of $P_2$: $n_2 \geq 2\ell' + 1$.

Assuming for contradiction that $n_2 \leq p$, we combine the bounds from the steps above: $2\ell' + 1 \leq n_2 \leq p \leq n_1 - n_2 \leq 2\ell'$.
This leads to the contradiction: $2\ell' + 1 \leq 2\ell' \implies 1 \leq 0$.
Thus, $n_2 > p$ always holds.
\end{proof}

The crucial algorithmic implication of \Cref{lem:structural}
is \Cref{lem:LCPs}: we can reduce the computation of $\MP$ to $\cO(1)$ LCP queries by accessing only $F$. Our block decomposition will be precisely based on this result.

\begin{lemma}\label{lem:LCPs}
    Let $S$ be a string of length $n$. Given read-only access to an arbitrary fragment $F=S[i\dd j]=B_1B_2B_3B_4$ of $S$ of length $\ell=4\ell'$, with $\ell'=|B_1|=|B_2|=|B_3|=|B_4|>0$, 
    all palindromes in $\MP$ can be computed in $\cO(|F|)$ time and space plus the time and space to answer $3$ LCP queries on string $S':=S\cdot S^R$. 
\end{lemma}

Before proving \Cref{lem:LCPs}, we prove some standard facts relating palindromes to periodicity.

    \begin{fact}\label{fct:ext}
    If string $P$ with a period $p$ is a palindrome and string $S=cPc'$, for two letters $c,c'$, has $p$ as a period too, then $S$ is a palindrome.    
    \end{fact}
    \begin{proof}
     Let $n := |P|$. Since $S = cPc'$ has a period $p$, $c = P[p-1]$ and $c' = P[n-p]$. Since $P$ is a palindrome, $P[p-1] = P[n-p]$.
     By combining these three results we have: $c = P[p-1] = P[n-p] = c'$.
     Since $c=c'$, and $P$ is a palindrome, the string $S=cPc'$ is also a palindrome, by the definition of a palindrome.
    \end{proof}
    \begin{fact}\label{fct:max}
    If string $P$ with a period $p$ is a palindrome, string $cP$, for a letter $c$, has $p$ as a period, and string $Pc'$, for a letter $c'$, does not have $p$ as period, then $S=cPc'$ is not a palindrome.
    \end{fact}
    \begin{proof}
Let $n := |P|$. By the statement, we have $c = P[p-1]$ and $c' \neq P[n-p]$.
Since $P$ is a palindrome, $P[i] = P[n-1-i]$, for all valid indices $i$. 
For the index $i=p-1$, this gives $P[p-1] = P[n-1-(p-1)] = P[n-p]$.
Combining these results, we get: $c = P[p-1] = P[n-p] \neq c'$.
This shows that the first and last letters of $S=cPc'$ are not equal ($c \neq c'$), and therefore $S$ is not a palindrome, by the definition of a palindrome.
\end{proof}

\begin{proof}[Proof of \Cref{lem:LCPs}]
    We compute $\PP$ in $\cO(|F|)$ time using Manacher's algorithm and apply \Cref{lem:structural}.

    If $|\PP|=0$, then $|\MP|=0$.
    If $|\PP|=1$, then we know that $|\MP|=1$, and proceed as follows.
    We find the center $c$ of $P\in\PP$ in $B_2$ and ask     
    $\textsf{LCP}_{S'}(c,2n-c-1)$, if $P$ is odd-length palindrome, or $\textsf{LCP}_{S'}(\lceil c \rceil,2n-\lceil c \rceil)$, if $P$ is even-length palindrome, to find the corresponding $M\in \MP$ with center $c$. 
    
    If $|\PP|\geq 2$, then by \Cref{lem:structural} and \Cref{fct:per-pal}, 
    every palindrome in $\PP$ has 
    period $p:=|P_1|-|P_2|$, 
    where $P_1$ and $P_2$ are the longest and second longest palindromes in $\PP$, respectively.
    We check how long this period $p$ extends to the left and to the right of $F$ 
    by asking two $\textsf{LCP}_{S'}$ queries. 
    For the left part, we ask $a:=\textsf{LCP}_{S'}(2n-i-p,2n-i)$; 
    and for the right part, we ask $b:=\textsf{LCP}_{S'}(i,i+p)$.
    The periodic fragment of $S$ with period $p$ is thus $S[i-a\dd i+b-1]$.
    For every $P\in \PP$ and $M\in \MP$ with the same center $c$, 
    we have two cases:
    \begin{description}
        \item \textbf{Case A}: $|P|=b-a$. The left and right endpoints of the periodic fragment are reached \emph{simultaneously} from the center $c$ of palindrome $P \in \PP$ with $|P|=b-a$. We ask $\textsf{LCP}_{S'}(c,2n-c-1)$ if $P$ is odd-length palindrome, or $\textsf{LCP}_{S'}(\lceil c \rceil,2n-\lceil c \rceil)$ if $P$ is even-length palindrome, to find $M\in\MP$ with center $c$. This is correct because $M$ is a palindrome by \Cref{fct:ext} and it is maximal by the LCP query definition.
        \item \textbf{Case B}: $|M|=\min(|P|+2a,2b-|P|)$. We reach one of the two endpoints of the periodic fragment \emph{first} coming from the center of $P$. If periodicity breaks first on the left, $|M|=|P|+2a$; otherwise,
        \begin{align*}
|M| &= (\text{left extension}) + |P| + (\text{right extension}) \\
&= (b - |P|) + |P| + (b - |P|) \\
&= 2b - |P|
\end{align*}
        This is correct because $M$ is a palindrome by \Cref{fct:ext}; and it is maximal by \Cref{fct:max}  and symmetry.
    \end{description}
\end{proof}

\begin{remark}
A lemma very similar to \Cref{lem:LCPs} was also shown by Charalampopoulos et al.~\cite{DBLP:conf/cpm/Charalampopoulos22}, who considered the LPS problem in a very different setting: string $S\in\Sigma^n$ is given in packed form and the task is to solve LPS sequentially in sublinear time for a small alphabet $\Sigma$. The authors showed a version of \Cref{lem:LCPs} in a special setting; that is, for $\ell':=\max(1,\frac{1}{8}\log_\sigma n)$, with $\sigma=|\Sigma|$, and for computing only a \emph{longest} palindrome in $\MP$. Our proof is also considerably simpler---precisely due to the fact that we have no constraint on $\ell'$.    
\end{remark}

\paragraph{The Algorithms.} Armed with the above combinatorial insights, we are now in a position to state our main algorithm. We decompose $S$ into a sequence of length-$\ell'$ \emph{blocks} with $\ell':=n^{1-\epsilon}$. These blocks do not overlap. We also split $S$ into a sequence of length-$\ell$ \emph{superblocks}, with $\ell:=4\ell'$, such that every two consecutive superblocks overlap by $3\ell'$ positions. We use $n^{\epsilon}$ machines, called \emph{block} machines. Each machine is assigned one such superblock. Thus, the memory used by every machine is $\cO(n^{1-\epsilon})$. The goal is to compute the maximal palindromes that are centered within the \emph{second block} of each superblock. Observe that any maximal palindrome centered at the first or last $2n^{1-\epsilon}$ positions of $S$ can be of length at most $4n^{1-\epsilon}=\ell$. Thus, the first and last machines can finish this assignment locally using Manacher's algorithm~\cite{DBLP:journals/jacm/Manacher75}. See \Cref{fig:block-decomposition} that illustrates this \emph{block decomposition}.

\begin{figure}[ht]
    \centering
    \includegraphics[width=0.85\linewidth]{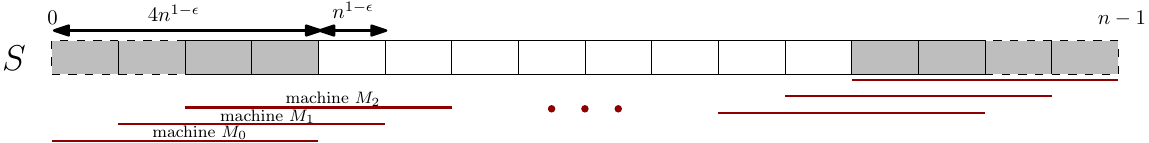}
    \caption{Illustration of the \emph{block decomposition}. The top (black) lines show the size of the (overlapping) superblocks and
    the size of the (non-overlapping) blocks. The bottom (dark red) lines show how the superblocks are assigned to the machines. Thus, we use $n^{\epsilon}$ machines with local memory $\cO(n^{1-\epsilon})$ for this decomposition.}
    \label{fig:block-decomposition}
\end{figure}

\Cref{lem:LCPs} leads to a simple and elegant algorithm because each block machine determines $\cO(1)$ LCP queries by \emph{accessing its local superblock}. This is a key difference from the algorithm by Gilbert et al., where a more complicated method involving ranges of blocks is used to determine the set of LCP queries. For every second block, denoted by $B_2$, in a superblock $F$, we will compute: (i) the maximal palindromes of $F$ with centers in $B_2$ computed locally; and (ii) the maximal palindromes of $S$ computed by applying \Cref{lem:LCPs} to $F$ with $\ell:=4n^{1-\epsilon}$. For (i), we use Manacher's algorithm on $F$ to compute all maximal palindromes.

Thus, what remains is to show how we can fully implement \Cref{lem:LCPs} in our setting; i.e., how we can answer the LCP queries.
For this, we will resort to a second (folklore) decomposition of string $S':=S\cdot S^R$, called \emph{modular decomposition}, which was also used by Gilbert et al. We make use of $n^\epsilon$ machines.
The machine $M_x$ with ID $x$, for all $x\in [0,n^{\epsilon})$, stores the KR fingerprint of every fragment of length $n^\epsilon$ starting at position $i$ in $S'$ so that $i \mod n^\epsilon = x$.
Thus, the memory used by every machine is $\cO(n^{1-\epsilon})$ and the total memory used is $\cO(n)$. The total time to compute all fingerprints is $\cO(n)$ by using \Cref{fct:fastKR}. The key is to apply \Cref{lem:LCPs} per machine, which determines the $\cO(1)$ LCP queries requested per machine. Even if all these queries are assigned to a single machine to answer them, their number would never exceed the $\cO(n^{1-\epsilon})$ local memory of that machine; indeed, recall $\epsilon \leq 0.5$. However, the total size of the messages for these requests may exceed the local memory of the machine, and so to balance these requests
we resort to \emph{data replication}.
Every machine (evenly) decomposes its data into as many fragments as the number of requested messages, and sends these fragments to as many different machines, which then replicate the data in as many copies. Since a machine's local memory is sufficient to store a message from every other machine in a single communication round, it is straightforward to arrange where to send and from where to retrieve the data. Thus each machine sends or receives data of $\cO(n^{1-\epsilon})$ size in $\cO(1)$ rounds. Moreover, the data replication takes $\cO(n)$ total time. See~\Cref{fig:replicate} that illustrates the data replication.

\begin{figure}[ht]
\begin{subfigure}{0.5\textwidth}
\includegraphics[width=0.9\linewidth]{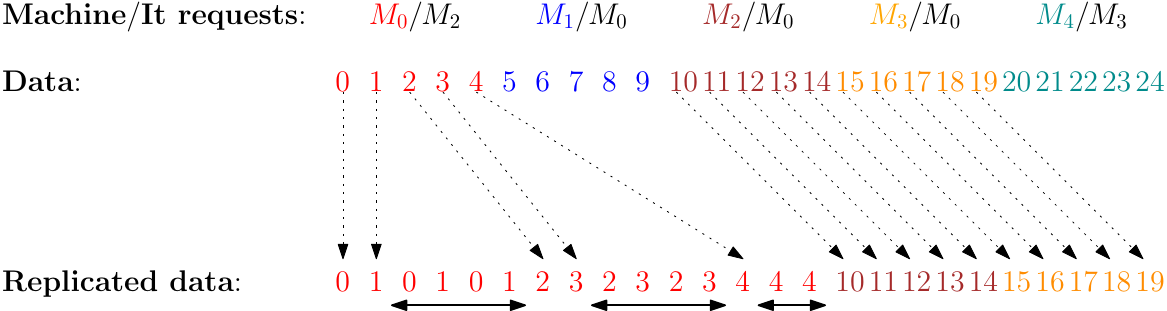} 
\caption{How data is replicated locally.}
\label{fig:replicate1}
\end{subfigure}
\begin{subfigure}{0.5\textwidth}
\includegraphics[width=0.9\linewidth]{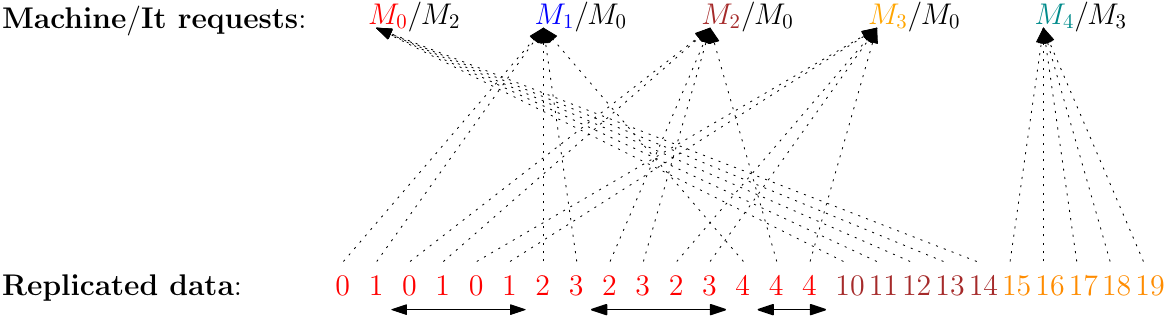}
\caption{How data is then received by the machines.}
\label{fig:replicate2}
\end{subfigure}
\caption{An example of the \emph{data replication} process with $n=25$ and $\epsilon=0.5$. Each machine sends or receives messages of $\cO(\sqrt{n})$ size in $\cO(1)$ rounds. Moreover, the data replication takes $\cO(n)$ total time.}
\label{fig:replicate}
\end{figure}

Armed with the above tools, we proceed to answering the LCP queries as follows.
For a single $\textsf{LCP}_{S'}(i,j)$ query, the machine with ID $x$ requests
the $\cO(n^{1-\epsilon})$ KR fingerprints of the fragments of length $n^{\epsilon}$
that compose $S'[i\dd]$ and $S'[j\dd]$ to be stored locally. This request can be implemented in $\cO(1)$ rounds, because, by the modular decomposition, the fingerprints are stored on at most two different machines: on machines with IDs $y=i \mod n^\epsilon$ and $z=j \mod n^\epsilon$; and to guarantee that no machine sends fingerprints that exceed its local memory, we employ data replication as described above.
By comparing these fingerprints in $\cO(n^{1-\epsilon})$ time, we determine the two length-$n^{\epsilon}$ fragments into which $S'[i\dd]$ and $S'[j\dd]$ no longer match. We can then compare these fragments letter by letter
by sending a request to the $\cO(1)$ \emph{block} machines that contain the two (non-hashed) fragments; again, to guarantee that no machine sends letters that exceed its local memory, we employ data replication. By \Cref{lem:LCPs}, $\cO(1)$ total rounds are required and, by \Cref{fct:KR}, the result is correct w.h.p. We have arrived at \Cref{the:main}.

\begin{theorem}\label{the:main}
For any $\epsilon \in (0,0.5]$,
there is an MPC algorithm
computing all maximal palindromes in $S\in\Sigma^n$ in $\cO(1)$ rounds. 
The total time and memory are $\cO(n)$, with $\cO(n^{1-\epsilon})$ memory per machine, w.h.p.
\end{theorem}

In the AMPC model, we can bypass the constraint $\epsilon \leq 0.5$ arising from modular decomposition and data replication, through a few simple modifications. This is a key difference from the algorithm by Gilbert et al., who need the constraint for their block decomposition too. Indeed, they solve LPS in the AMPC model by reducing it to suffix tree construction, which induces polylogarithmic factors and is significantly more involved.

\begin{corollary}\label{coro:main}
For any constant $\epsilon \in (0,1)$,
there is an AMPC algorithm computing all maximal palindromes in $S\in\Sigma^n$ in $\cO(1)$ rounds. The total time and memory are $\cO(n)$, with $\cO(n^{1-\epsilon})$ memory per machine, w.h.p.
\end{corollary}
\begin{proof}
We follow the algorithm underlying \Cref{the:main}, with $n^{\epsilon}$ machines, up to determining the $\cO(1)$ LCP queries on string $S':=S\cdot S^R$ per machine. Instead of the modular decomposition, we proceed as follows. Let $s:=n^{1-\epsilon}$.
Using a prefix-sum algorithm~\cite{DBLP:conf/isaac/GoodrichSZ11} augmented with~\Cref{fct:fastKR}, we can store $\phi(S'[0\dd i])$, for all $i\in[0,2n)$, in the shared memory using $\cO(\log_{s} n)=\cO(1/(1-\epsilon))=\cO(1)$ rounds and $\cO(n\log_{s} n)=\cO(n/(1-\epsilon))=\cO(n)$ total time and memory. We can then compute the KR fingerprint of any arbitrary fragment of $S'$ in $\cO(1)$ time w.h.p., because, by~\Cref{fct:fastKR}, we can compute $\phi(S'[i\dd j])$
using $\phi(S'[0\dd i-1])$ and $\phi(S'[0\dd j])$.
For any two suffixes of $S'$, we binary search \emph{adaptively} to find their LCP value. Since every machine asks $\cO(1)$ LCP queries, the total time and memory is $\cO(n+n^{\epsilon}\log n)=\cO(n)$, with $\cO(n^{1-\epsilon}+\log n)=\cO(n^{1-\epsilon})$ memory per machine.    
\end{proof}

\section{Concluding Remarks}\label{sec:con}

The purpose of this paper is twofold. From the results perspective, our algorithms may inspire further work on classical string matching problems in MPC or AMPC, and they may serve as a useful guide for parallel implementations. From a pedagogical standpoint, our paper may serve as material for a class on parallel and distributed algorithms showcasing a remarkably simple and optimal algorithm obtained through combinatorial insight. We propose the following two questions for future work:

\begin{itemize}
    \item Is there an efficient \emph{deterministic} algorithm for the LPS problem in the MPC model?
    \item Is there an efficient algorithm for the LPS problem in the MPC model \emph{for any} $\epsilon \in (0,1)$?
\end{itemize}

\section*{Acknowledgments} This work was supported by the PANGAIA and ALPACA projects that have received funding from the European Union’s Horizon 2020 research and innovation programme under the Marie Skłodowska-Curie grant agreements No 872539 and 956229, respectively.

\bibliographystyle{alpha}
\bibliography{references}
\end{document}